\documentclass{article}
\usepackage{makecell}

\usepackage{amsmath}

\usepackage{PRIMEarxiv}

\usepackage[utf8]{inputenc} 
\usepackage[T1]{fontenc}    
\usepackage{hyperref}       
\usepackage{url}            
\usepackage{booktabs}       
\usepackage{amsfonts}       
\usepackage{nicefrac}       
\usepackage{microtype}      
\usepackage{lipsum}
\usepackage{fancyhdr}       
\usepackage{graphicx}       
\graphicspath{{media/}}     

\usepackage{subfig}
\DeclareMathOperator*{\argmax}{arg\,max}

\usepackage{amsthm}
\usepackage{amssymb}
\newtheorem{theorem}{Theorem}
\usepackage{blkarray}
\usepackage{multirow}
\usepackage{makecell}

\title{Limits of Predictability in Top-N Recommendation
}

\author{
	En Xu$^{a}$, Zhiwen Yu$^{a,1}$, Ying Zhang$^{a}$, Bin Guo$^{a}$, Lina Yao$^{b}$
	\thanks{\textsuperscript{1}To whom correspondence should be addressed. E-mail: zhiwenyu@nwpu.edu.cn (Z. Y.)}\\
  $^a$ School of Computer Science, Northwestern Polytechnical University, Xi'an 710129, China\\
  $^b$ School of Computer Science and Engineering, University of New South Wales, Sydney 2052, Australia\\
}

\makeatletter
\def\thanks#1{\protected@xdef\@thanks{\@thanks
        \protect\footnotetext{#1}}}
\makeatother

\begin{document}
\maketitle

\begin{abstract}
\textit{Top-N} recommendation aims to recommend each consumer a small set of $N$ items from a large collection of items, and its accuracy is one of the most common indexes to evaluate the performance of a recommendation system. While a large number of algorithms are proposed to push the \textit{Top-N} accuracy by learning the user preference from their history purchase data, a \textbf{predictability} question is naturally raised - whether there is an upper limit of such \textit{Top-N} accuracy. This work investigates such predictability by studying the degree of regularity from a specific set of user behavior data. Quantifying the predictability of \textit{Top-N} recommendations requires simultaneously quantifying the limits on the accuracy of the $N$ behaviors with the highest probability. This greatly increases the difficulty of the problem. To achieve this, we firstly excavate the associations among $N$ behaviors with the highest probability and describe the user behavior distribution based on the information theory. Then, we adopt the Fano inequality to scale and obtain the \textit{Top-N} predictability. Extensive experiments are conducted on the real-world data where significant improvements are observed compared to the state-of-the-art methods. We have not only completed the predictability calculation for $N$ targets but also obtained predictability that is much closer to the true value than existing methods.
We expect our results to assist these research areas where the quantitative requirement of \textit{Top-N} predictability is required.
\end{abstract}

\keywords{Predictability, \textit{Top-N} Recommendation, Recommender Systems, Information theory, Statistics and Probability}

\section{Introduction}\label{sec:introduction}
A recommender system is a kind of personalized information filtering technology used to recommend items in line with his interests to a specific user. Recommender systems have been successfully applied in many fields, including e-commerce, information retrieval, social networks, location services, a news feed, and other areas \cite{liu2010personalized}.  
There are many scenarios in the recommender system, such as online shopping, where the platform can recommend $N$ items to the user simultaneously, and the recommendation is successful as long as the clicking or buying behavior occurs. Therefore, \textit{Top-N} accuracy is more often used as an evaluation metric in recommender systems.
There are many outstanding algorithms to improve the accuracy of recommendation, including collaborative filtering \cite{sarwar2001item}, content-based recommendation algorithms \cite{pazzani2007content}, and deep learning-based recommendation algorithms \cite{shan2016deep}. Research on these algorithms still does not answer the limit of accuracy that can be achieved on the dataset. Therefore, we need to research the predictability of \textit{Top-N} recommendations. Predictability ($\Pi$) refers to the maximum accuracy that the optimal algorithm can achieve given the dataset. The measure of predictability gives us an idea of the degree to which a user's behavior is regular. Predictability also allows us to understand the extent to which the field is currently evolving.
However, all the existing work can only calculate the \textit{Top-1} predictability \cite{song2010limits}, and there is no theory to quantify the \textit{Top-N} predictability at present.
\begin{figure*}[!t]
	\centering
	\setlength{\abovecaptionskip}{0pt}
	\setlength{\belowcaptionskip}{0pt}
	\includegraphics[width=0.8\textwidth]{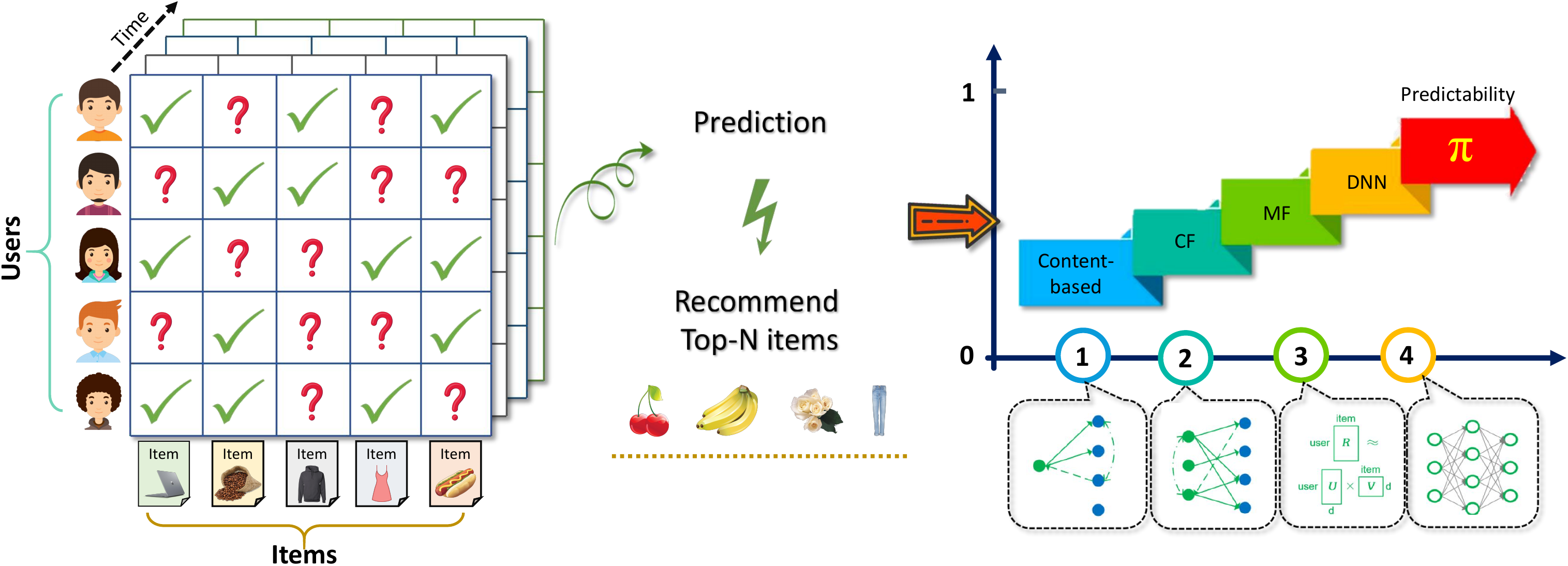} 
	\caption{Schematic representation of the predictability of \textit{Top-N} recommendation. As new algorithms continue to be proposed, the accuracy of \textit{Top-N} recommendation continues to improve. Predictability is the potential maximum of accuracy.}
	\label{fig_framework}
	\vspace{0pt}
\end{figure*}

Song et al. \cite{song2010limits} quantified the uncertainty of human movement data through information entropy and calculated that the limits of predictability of human movement behavior was 93\% by using inequality scaling. Zhou et al. \cite{L2015Toward} randomly extract a small number of links from the network. The influence on the eigenvector space of the network is small, which indicates that the network is regular. A structural consistency index is proposed to measure link predictability. Sun et al. \cite{sun2020revealing} found that the compression length of the sequence was correlated with the difficulty of prediction, thus exploring the predictability of network structure. Only Song's theory \cite{song2010limits} deduced the limits on predictability. Other works found approximate correlations between other vital indicators and predictability, which could not guarantee that the calculated predictability must be greater than the real predictability. However, the existing methods all solve the quantification problem of the \textit{Top-1} predictability and fail to obtain the \textit{Top-N} predictability. There are two things that we need to improve in our current work. Firstly we need to give a theory for the computation of \textit{Top-N} predictability. Secondly, our approach should derive the limits on predictability rather than just proposing new metrics that have an approximate relationship with predictability.

The existing theories to derive the limits of the \textit{Top-1} predictability adopt the method of information theory. The overall entropy $S_{whole}$ of user behavior can be calculated by certain calculation methods. By scaling with Fano inequality, the limits of the behavior entropy $S_{tail}$ can be calculated for all behavior except the behavior with the highest probability ($\Pi$). After calculating $S_{whole}$ and $S_{tail}$, according to the inequality among $S_{whole}$, $S_{tail}$ and $\Pi$, the limits of $\Pi$ is deduced. But if we want to solve for the predictability of the most likely $N$ behaviors simultaneously, that's not going to work.
There will be multiple unknown variables in a formula after scaling. 
To solve this problem, we use group behavior data to quantify the ratio of predictability among the $N$ items with the highest probability in \textit{Top-N} recommendation, obtain the correlation of $N$ unknown variables, and add new constraints.
Therefore, we can convert several unknown variables into $1$ to achieve the \textit{Top-N} predictability measurement work.
At the same time, our method has other important properties. The discrepancy between the predictability we derive and the actual predictability is minor. Finally, we also quantified the error of solving the predictability under different scaling scenarios, achieving a more accurate estimate of the predictability. In summary, this paper makes the following contributions:
\begin{itemize}
	\item To the best of our knowledge, this work is the first to formally define the predictability in \textit{Top-N} recommendation and give a solution. We can quantify the predictability of the $N$ behaviors with the highest probability, whereas existing works provide only the one with the highest probability.
	\item Our method can quickly calculate the highest accuracy under a specific dataset, which helps judge the difficulty of the problem and the improvement space of existing methods. It is simple compared to implementing specific algorithms to obtain accuracy.
	\item Our study not only achieves the quantification of \textit{Top-N} predictability but also dramatically improves the accuracy of quantification, which makes the research on predictability more valuable for application.
\end{itemize}

\section{Related Works}
\subsection{Top-N recommendation}
Due to the tremendous academic and commercial value of \textit{Top-N} recommendations, research on recommender systems has been in full swing in both academia and industry over the past decades. We briefly review representative work in this area. The content-based recommendation algorithms were heavily studied in the early days, and the core idea was to calculate the similarity between users and items by their attributes \cite{pazzani2007content}. Based on the similarity between users and items, the most relevant \textit{Top-N} items are finally recommended to the target users. The content-based recommendation algorithm only needs to calculate the similarity between the target users and the items individually based on specific rules. The method is not very computationally intensive. However, feature extraction is often difficult. Instead of collecting rich user attributes, collaborative filtering directly establishes associations based on user-item interaction records \cite{su2009survey}. The core idea of collaborative filtering is that similar users have similar behavioral preferences, and similar items will be interacted by similar users. Based on the former, user-based collaborative filtering is proposed, and based on the latter \cite{zheng2010collaborative}, item-based collaborative filtering is developed \cite{sarwar2001item}. The model-based collaborative filtering uses the interaction data of users and items as a whole to learn the global model to predict the missing interaction behaviors \cite{pennock2013collaborative}.

With the successful development of deep neural networks, many neural recommendation models have emerged in recent years. Deep networks have achieved significant performance improvements in handling large sets of user items and complex interactions between users and items due to their excellent representation learning capability and powerful fitting ability. Rating prediction algorithms based on various deep learning models are also emerging, such as restricted Boltzmann machines, deep belief networks, autoencoders, recurrent neural networks, convolutional neural networks, multilayer perceptrons, etc. \cite{salakhutdinov2007restricted,quadrana2017personalizing,man2017cross}. DNNs can model users' interests by directly inputting historical behaviors due to their end-to-end advantages \cite{covington2016deep}. RNNs can capture the evolution of users' interests in time-series behaviors \cite{quadrana2017personalizing}. GNNs can learn the representation of users and items by using the powerful performance of processing graph data \cite{fan2019graph}. There is not only the direct use of deep learning models to improve performance but also a lot of work to further develop typical recommendation algorithms with the help of deep learning. NeuMF is a classical collaborative filtering framework that uses multilayer perceptrons to model user-item interactions \cite{he2017neural} accurately. Wide\&Deep \cite{cheng2016wide} and DeepFM \cite{guo2017deepfm} both integrate feature learning and deep modeling and include both user interest breadth and depth modules. Clever algorithms are constantly proposed, and the accuracy of \textit{Top-N} recommendations is increasing, but this still does not answer our question, what is the maximum accuracy? Therefore we need further research on predictability.

\subsection{Predictability}
Predictability research aims to uncover the regularity of the actor's subject, which can reflect the inherent fundamental properties of the subject. At the same time, predictability allows us to know the field's current state of development and determine the feasibility of achieving the corresponding accuracy for specific problems. Early studies of predictability focused on turbulent systems, quantifying the predictability of the climate system by analyzing the dynamical equations established for the system to obtain the effect of initial condition uncertainty and boundary uncertainty on error growth. In \textit{Science 2010}, Song analyzed human mobility data from the perspective of information theory, quantified the chaos of behavior using entropy, and finally obtained the predictability of human mobility behavior up to 93\% using the scaling of Fano's inequality \cite{song2010limits,1961242}. Due to the generalizability of the theory, it has been widely applied to other scenarios such as human communication sequences \cite{zhao2013emergence}, vehicle mobility \cite{wang2015predictability}, IP address sequences for cyber attacks \cite{chen2015spatiotemporal}, stock price changes \cite{fiedor2014frequency}, electronic health records \cite{dahlem2015predictability}, and so on. Some work further digs into what factors significantly affect predictability through Song's theory, and finally obtains some important factors of predictability, such as Spatio-temporal resolution \cite{lin2012predictability}, exploration preferences \cite{cuttone2018understanding}, and data quality \cite{iovan2013moving}. Some works try different quantification methods of entropy to analyze the confusion of data from different aspects, such as mutual information \cite{chen2016temporal}, instantaneous entropy \cite{baumann2013use}, alignment entropy \cite{scarpino2019predictability}. Smith et al. find a more accurate set of candidate locations for the next human moment by topological constraints on geographic space and get more accurate predictability \cite{smith2014refined}. Sun et al. further compressed and coded the sequence by converting the graph into a sequence and obtained the correlation between the shortest compression length and the predictability of the graph \cite{sun2020revealing}. However, the existing methods all aim to explore the predictability of \textit{Top-1} and cannot be directly applied to the predictability of \textit{Top-N} recommendations.

\subsection{Predictability of recommender systems}
Only a tiny amount of work has been done to study the predictability of recommender systems. Related jobs can be roughly divided into three categories: direct applications of Song's theory to recommender systems; the limits on accuracy proposed for the shortcomings of specific methods; and mining regularity results on recommender datasets. Both Krumme \cite{krumme2013predictability}, and Jarv's \cite{jarv2019predictability} works directly use the theory of predictability of human movement behavior to derive the predictability of human consumption behavior in recommendation scenarios. Meanwhile, Jarv roughly measures the predictability of the recommender system by counting the number of first-time samples in the test set \cite{jarv2019predictability}. As for the diffusion-based method, Zhang et al. pointed out that if the connection could not be established through diffusion within a certain number of steps, it would be impossible to make further correct recommendations \cite{zhang2019predictability}. Based on this analysis, the accuracy limits of the method is further deduced. Alex \cite{krumme2013predictability} analyzed the regularity of customers' visiting patterns by recording users' electronic consumption behaviors and found that although consumers have different personal preferences, over time, each person's visiting patterns of merchants have a high regularity. Users are less predictable in the short run but more predictable in the long run. However, most of the existing studies focused on the predictability of \textit{Top-1}, and failed to obtain the predictability of multiple candidate targets, namely \textit{Top-N} recommendation.


\section{THE PROPOSED METHOD}
In this section, we first give the existing methods for computing the predictability of recommender systems. Then we introduce our proposed method for predictability in \textit{Top-N} recommendations, and finally, we demonstrate in detail the excellent properties of our approach.
\subsection{Predictability of recommender system}
In the recommender system, let $U=\{u_{1},u_{2},\cdot\cdot\cdot,u_{|U|}\}$ represent a group of users, $V=\{v_{1},v_{2},\cdot\cdot\cdot,v_{|V|}\}$ represent a group of items, and list $B_{u}=\{v_{1}^{u},\cdot\cdot\cdot,v_{t}^{u},\cdot\cdot\cdot,v_{n_{u}}^{u}\}$ represent the behavior sequence of user $u \in U$, where $v_{t}^{u}\in V$ is the item clicked or purchased by $u$ at time step $t$, and $n_{u}$ is the length of user's behavior sequence. 
Given a dataset ($\mathcal{D}$), the highest accuracy that any algorithm can achieve on the dataset ($\mathcal{D}$) is predictability ($\Pi$).
$\Pi^{max}$ represents the limits on predictability. Through scaling of Fano's inequality, Song et al. \cite{song2010limits} finally deduced that if a user has clicked or purchased $M$ different items and the entropy of the behavior sequence is $S$, then the limit of predictability of the user $\Pi^{max}$ can be obtained by Eq. \ref{PreScience}. 
\begin{align}
	\label{PreScience}
	S = -\Pi^{max}\log_2\Pi^{max} - (1-\Pi^{max})\log_2(1-\Pi^{max}) + (1-\Pi^{max})\log_2(M-1),			
\end{align}
where the expression of entropy $S$ and its calculation method as shown below
\begin{align}
	\label{RealEntropy}
	S_{real} = - \sum_{T'\in T} P(T')	\log_2[P(T')].			
\end{align}

The real entropy $S_{real}$ not only depends on the frequency of purchases but also on the order of purchasing items. From the formula, we can know that the uncertainty contained in the same behavior $v_{i}$ in the user sequence is different if the time of occurrence is different. Therefore, the real entropy quantifies the chaotic degree of temporal sequence behavior.
Finding all subsets of a given set has exponential complexity $(O(2^{n}))$. We used the Lempel-Ziv estimator to calculate the actual entropy. The Lempel-Ziv estimator can quickly converge to the actual entropy. For user behavior sequences, $S_{real}$ can be estimated in the following ways
\begin{align}
	\label{formula_5}
	S^{est} = \left( \frac{1}{n} \sum_{i}\Lambda_{i} \right)^{-1} \ln n,		
\end{align}
where $\Lambda_{i}$ represents the length of the shortest substring starting at position $i$, which has not previously appeared from places $1$ to $i-1$.

\subsection{Defects of the existing method}
The existing method can calculate the predictability of items with the highest probability, that is, \textit{Top-1} predictability. We can try to treat the \textit{Top-N} candidates as a whole to obtain \textit{Top-N} predictability. Unfortunately, that idea doesn't work.
There are two variables that are important to evaluate in Eq. \ref{PreScience}, one is the entropy $S$ of the sequence and the other is the size of the candidate items set $M$ at the next moment. Considering the \textit{Top-N} candidate items as a whole to calculate the \textit{Top-N} predictability is actually equivalent to changing $M$ to $M-N$ and substituting it into the calculation, which corresponds to Eq. \ref{TOPN-PreScience}.
\begin{align}
	\label{TOPN-PreScience}
	S = -\Pi^{max}\log_2\Pi^{max} - (1-\Pi^{max})\log_2(1-\Pi^{max}) + (1-\Pi^{max})\log_2(M-N).		
\end{align}

Figure \ref{fig_M} shows the effect of $M$ on predictability. It can be seen that different $M$'s have an effect on predictability (Figure \ref{fig_M}, left), but that a mere small change in $M$ has little effect on predictability (Figure \ref{fig_M}, right). The figure on the left shows the effect of a change in $M$ on predictability at different entropies, with a significant change in predictability when $M$ changes by an order of magnitude. The figure on the right shows that there is little change in predictability when $M$ changes from 100 to 110. Whereas in a recommender system where the set of candidates $M$ is much larger than 100, a small change in $M$ has a much smaller impact on predictability. Therefore the method of calculating predictability by changing $M$ to $M-N$ will fail. $\Pi(\textup{Top-1})	\approx \Pi(\textup{Top-2}) \approx \dots \approx \Pi(\textup{Top-N})$. The \textit{Top-1} to \textit{Top-N} predictability obtained based on the existing method is the same, which is obviously inconsistent with the actual situation.
\begin{figure}[t]
	\centering
	\subfloat{
		\includegraphics[width=0.35\textwidth]{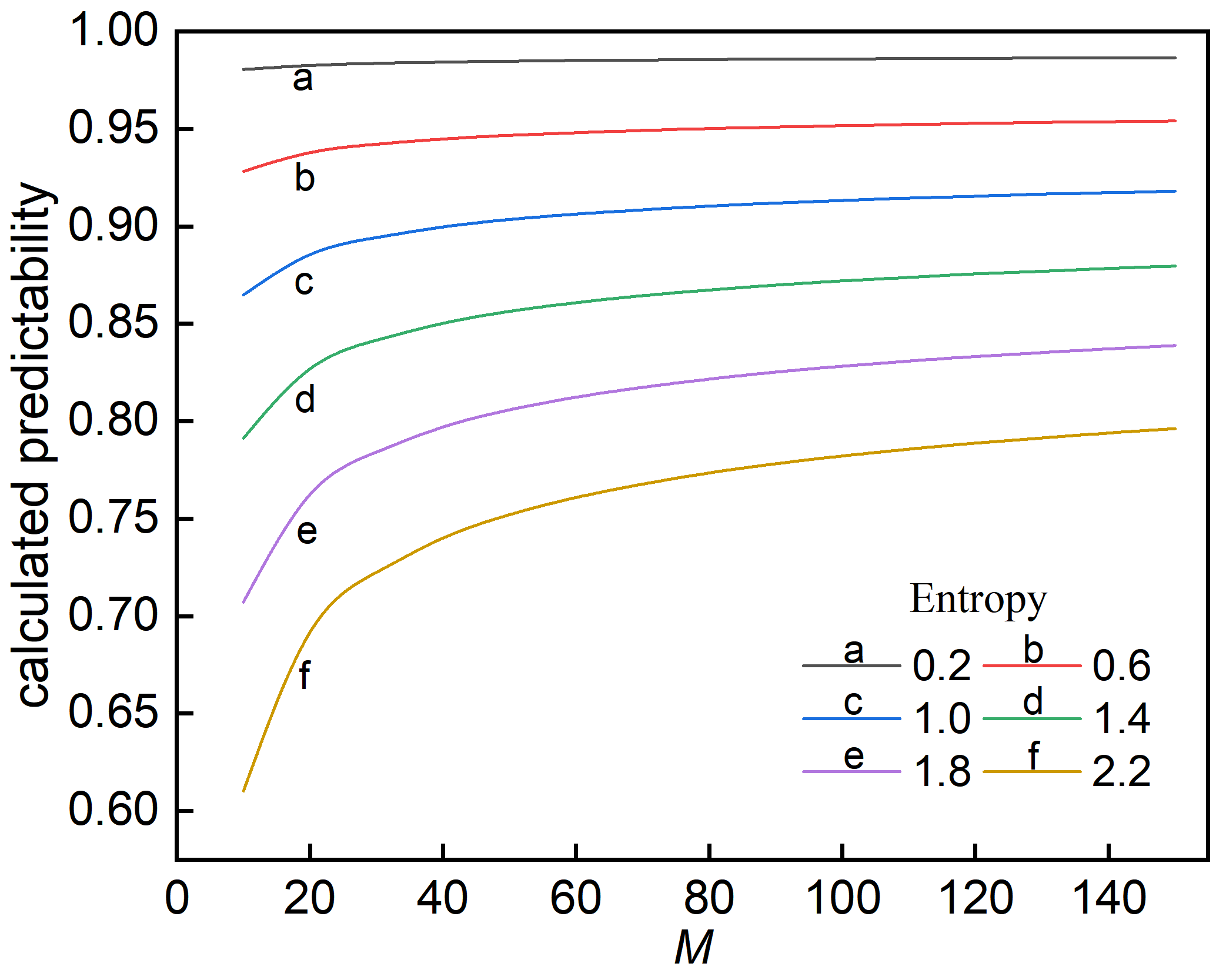}
	}
	\hspace{5mm}
	\subfloat{
		\includegraphics[width=0.361\textwidth]{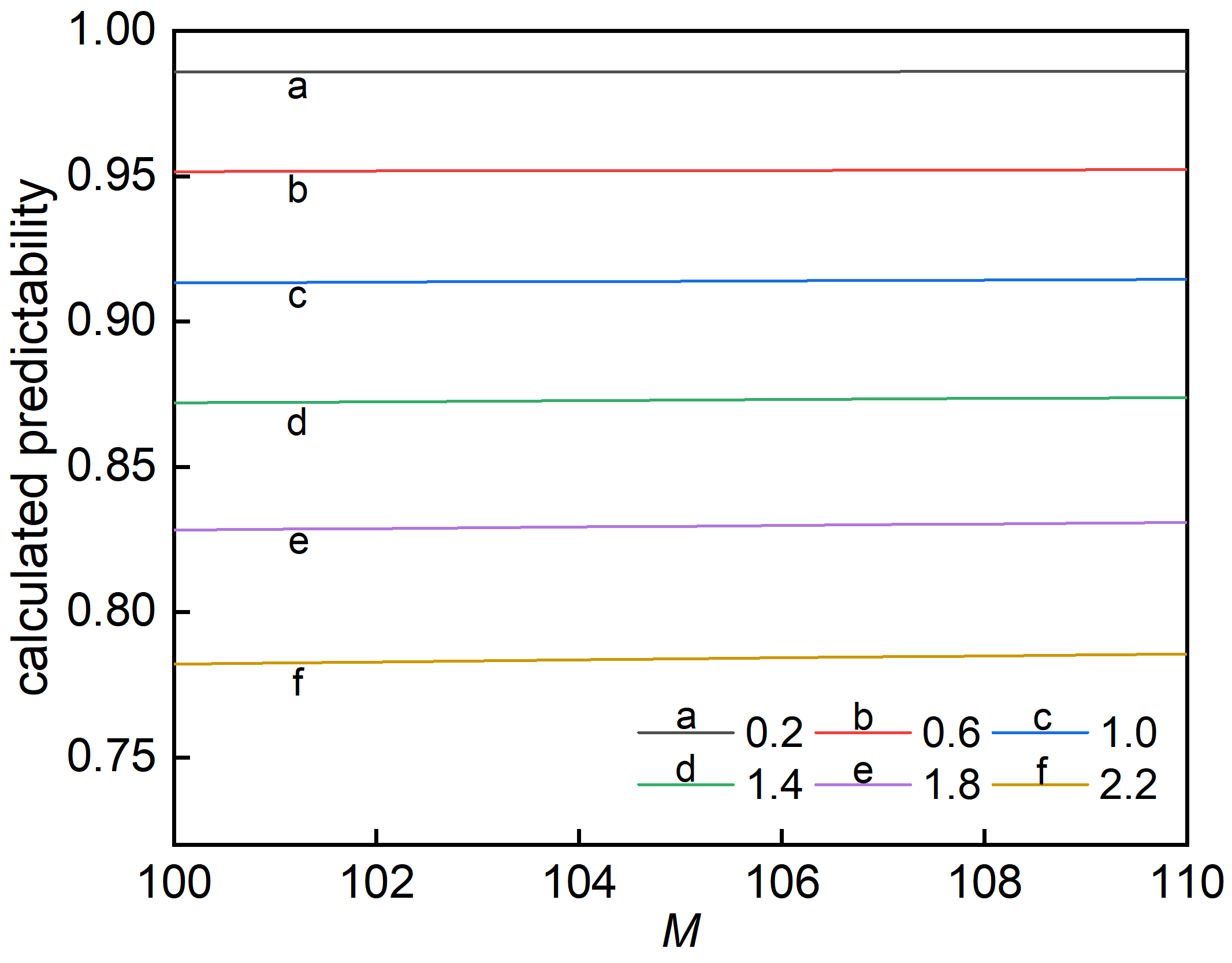}
	}
	\caption{The impact of M on predictability}
	\label{fig_M}
\end{figure}

\subsection{The limits on \textit{Top-N} predictability}
With the above predictability theory as a foundation, we next introduce the derivation of \textit{Top-N} predictability. Let's start by describing some of the symbols of predictability. We define $\pi^{u}_{1}(h_{n_{t}-1})$ as the probability of the most likely behavior of the user at the next moment given the historical behavior $h_{n_{t}-1}$. The subscript $1$ of $\pi^{u}_{1}(h_{n_{t}-1})$ indicates the maximum probability of the user's candidate behavior. Generally, subscript $i$ means the $i^{th}$ highest probability of the candidate behavior. Thus there are:
\begin{align}
	\pi^{u}_{1}(h_{n_{t}-1})=\sup_{x}\left\{Pr[X_{n_{t}}=x|h_{n_{t}-1}] \right\}
\end{align}

We define $\Pi^{u}_{i}(n_{t})$ as the predictability of user $u$ given a historical behavior of length $n_{t}-1$. The subscript $i$ indicates the predictability of the behavior with the $i^{th}$ highest probability. Let $P(h_{n_{t}-1})$ be the probability of observing a specific sequence of behaviors of length $n_{t}-1$. Then the predictability satisfies the following equation:
\begin{align}
	\Pi^{u}_{i}(n_{t}) = \sum_{h_{n_{t}-1}}P(h_{n_{t}-1})\pi^{u}_{i}(h_{n_{t}-1})
	\label{eq:2}
\end{align}

The total predictability of a user $\Pi^{u}_{i}$ can be obtained by summing the predictability of a user over the whole time as follows:
\begin{align}
	\Pi^{u}_{i} = \lim\limits_{n_{t}\to \infty}\frac{1}{n_{t}}\sum_{j=1}^{n_{t}}\Pi^{u}_{i}(j)  
\end{align}

By summing the predictability of all users under a dataset, the predictability $\Pi_{i}$ corresponding to the whole data set can be obtained.
\begin{align}
	\Pi_{i} = \lim_{n_{u},n_{t} \to \infty}\frac{1}{n_{u}}\sum_{j=1}^{n_{u}}\frac{1}{n_{t}}\sum_{k=1}^{n_{t}}\Pi^{j}_{i}(k)  
\end{align}

To derive the predictability of \textit{Top-N} recommendation, different from the previous method, which only reserved the user?s maximum probability behavior with all other behaviors scaled, we quantify all the actions of the user at the next moment.
We define $\{p_{1},p_{2},\cdot\cdot\cdot,p_{M}\}$ as the order of the probability of the user?s possible behaviors at the next moment from the largest to the smallest, and the corresponding behaviors is $\{v_{f_{1}},v_{f_{2}},\cdot\cdot\cdot,v_{f_{M}}\}$.
Therefore, we know that the entropy of user behavior sequence is:
\begin{align}
	S^{u} 	&= -\sum_{i=1}^{M}p_{i}\log_2p_{i} 	
\end{align}

Unfortunately, although the theoretical entropy of behavior can be listed, the above equation cannot be solved because only $S^{u}$ and the size of $M$ can be calculated from the data, and there are still $M$ unknown variables in the above equation. We define the probability between the $r$ items with the highest probability under the corresponding data to satisfy the following equation:
\begin{align}
	p_{i} = c_{i}p_{1}, \{1\leq i \leq r,i \in N^{*}\}
\end{align}
\begin{figure*}[t]
	\centering
	\setlength{\abovecaptionskip}{0pt}
	\setlength{\belowcaptionskip}{0pt}
	\includegraphics[width=0.99\textwidth]{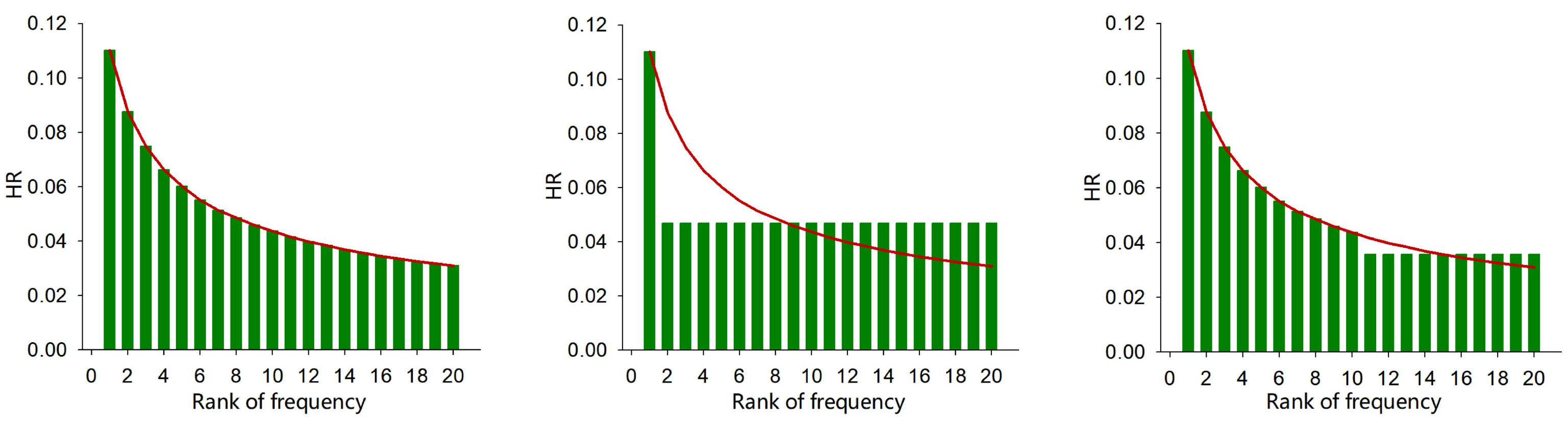} 
	\caption{Schematic diagram of the different scaling forms. The figure on the left shows the true probability distribution of the user's favorite items. The figure in the middle is a scaled version of the existing methods used to calculate the \textit{Top-1} predictability. The figure on the right is the scaled form we used to calculate the \textit{Top-N} predictability. Our scaled distribution is closer to the real distribution, so our calculated limits on the predictability is closer to the predictability.}
	\label{fig_zoom}
	\vspace{0pt}
\end{figure*}

We're going to scale the entropy of the remaining $M-r$ behaviors. The true distribution is $P(X|h)=\left(p_{1},p_{2},\cdot\cdot\cdot,p_{M} \right)$, $P'(X|h) = \left( p_{1} , p_{2} , \cdot\cdot\cdot , p_{r} , \frac{1-\sum_{i=1}^{r}p_{i}}{M-r} , \cdot\cdot\cdot , \frac{1-\sum_{i=1}^{r}p_{i}}{M-r} \right) $ is the new distribution constructed. A schematic of the scaling is shown in Fig. \ref{fig_zoom}. According to the maximum entropy theorem \cite{1957Information}, the entropy of the new distribution is greater than that of the original distribution. Therefore, we can get the following formula:
\begin{align}
	S^{u}(X_{n_{t}}|h_{n_{t}-1}) &\leq S^{u}(X'_{n_{t}}|h_{n_{t}-1})\\	
	&= -p_{1}\log_2p_{1}-p_{2}\log_2p_{2}-\cdot\cdot\cdot-p_{r}\log_2p_{r} - \sum \frac{1-\sum_{i=1}^{r}p_{i}}{M-r}\log_2\frac{1-\sum_{i=1}^{r}p_{i}}{M-r} \\	
	&= -p_{1}\log_2p_{1}-p_{2}\log_2p_{2}-\cdot\cdot\cdot-p_{r}\log_2p_{r}- (1-{\sum}_{i=1}^{r}p_{i})\log_2(1-{\sum}_{i=1}^{r}p_{i}) \notag\\	
	&\quad\,+ (1-{\sum}_{i=1}^{r}p_{i})\log_2(M-r)\\	
	&= S_{F_{r}}(p_{1},p_{2},...,p_{r}) \\
	&= S_{F_{r}}(\pi^{u}_{1}(h_{n_{t}-1}),\pi^{u}_{2}(h_{n_{t}-1}),...,\pi^{u}_{r}(h_{n_{t}-1}))
\end{align}	

If $S_{F_{r}}(\cdot)$ is regarded as a multivariate function, it is concave for $\pi^{u}_{i}(h_{n_{t}-1})$ and decreases monotonically with $\pi^{u}_{i}(h_{n_{t}-1})$ (see \emph{theorem \ref{theorem:1}}). So the $S_{F_{r}}(\cdot)$ function has two important properties where the Eq. (\ref{eq:12}) is Jensen's inequality for concave functions.
\begin{gather}
	S_{F_{r}}\left((a+b)/2 \right) \ge (S_{F_{r}}(a)+S_{F_{r}}(b))/2 \label{eq:12} \\ 
	(S_{F_{r}}(a)-S_{F_{r}}(b))(a-b) \le 0 
\end{gather}

We define the conditional entropy $S^{u}(X_{n_{t}}|h_{n_{t}-1})$ just like Eq. (\ref{eq:2}). From the above derivation, we can further find that $S^{u}(n_{t})$ should meet the following requirements:
\begin{small} 
	\begin{align}
		S^{u}(n_{t}) &= \sum_{h_{n_{t}-1}}P(h_{n_{t}-1})S^{u}(X_{n_{t}}|h_{n_{t}-1}) \notag \\
		&\le \sum_{h_{n_{t}-1}}P(h_{n_{t}-1})S_{F_{r}}(\pi^{u}_{1}(h_{n_{t}-1}),	 \pi^{u}_{2}(h_{n_{t}-1}),...,\pi^{u}_{r}(h_{n_{t}-1})) \notag \\
		&\le S_{F_{r}}(\sum_{h_{n_{t}-1}}P(h_{n_{t}-1})(\pi^{u}_{1}(h_{n_{t}-1}), \pi^{u}_{2}(h_{n_{t}-1}),...,\pi^{u}_{r}(h_{n_{t}-1})) ) \notag \\
		&= S_{F_{r}}(\Pi^{u}_{1}(n_{t}),\Pi^{u}_{2}(n_{t}),...,\Pi^{u}_{r}(n_{t}))		
	\end{align}
\end{small}

Using the conclusion of the above equation, we can further use Jensen's inequality \cite{1974Real} to obtain the correlation between entropy and predictability under the whole dataset.
\begin{small} 
	\begin{align}
		S &= \lim_{n_{u},n_{t} \to \infty}\frac{1}{n_{u}}\sum_{i=1}^{n_{u}}\frac{1}{n_{t}}\sum_{j=1}^{n_{t}}S^{i}(j) \notag\\
		&\le \lim_{n_{u},n_{t} \to \infty}\frac{1}{n_{u}}\sum_{i=1}^{n_{u}}\frac{1}{n_{t}}\sum_{j=1}^{n_{t}}S_{F_{r}}(\Pi^{i}_{1}(j),	\Pi^{i}_{2}(j),...,\Pi^{i}_{r}(j)) \notag \\
		&\le \lim_{n_{u} \to \infty}\frac{1}{n_{u}}\sum_{i=1}^{n_{u}}S_{F_{r}}(\lim_{n_{t} \to \infty}\frac{1}{n_{t}}\sum_{j=1}^{n_{t}}(\Pi^{i}_{1}(j),\Pi^{i}_{2}(j),...,\Pi^{i}_{r}(j))) \notag \\
		&\le S_{F_{r}}(\lim_{n_{u},n_{t} \to \infty}\frac{1}{n_{u}}\sum_{i=1}^{n_{u}}\frac{1}{n_{t}}\sum_{j=1}^{n_{t}}(\Pi^{i}_{1}(j),	\Pi^{i}_{2}(j),...,\Pi^{i}_{r}(j))) \notag \\
		&= S_{F_{r}}(\Pi_{1},\Pi_{2},...,\Pi_{r}) \notag \\
		&= S_{F_{r}}(c_{1}\Pi_{1},c_{2}\Pi_{1},...,c_{r}\Pi_{1})
	\end{align}
\end{small}

Now we define $\Pi^{max}_{1}$ as the solution to the following equation:
\begin{align}
	S 	&= S_{F_{r}}(\Pi^{max}_{1})		\notag	\\ 
	&= -c_{1}\Pi^{max}_{1}\log_2c_{1}\Pi^{max}_{1}-c_{2}\Pi^{max}_{1}\log_2c_{2}\Pi^{max}_{1} 	-\cdot\cdot\cdot-c_{r}\Pi^{max}_{1}\log_2c_{r}\Pi^{max}_{1}		\notag	\\	
	&\quad\,- (1-{\sum}_{i=1}^{r}c_{i}\Pi^{max}_{1})\log_2(1-{\sum}_{i=1}^{r}c_{i}\Pi^{max}_{1}) + (1-{\sum}_{i=1}^{r}c_{i}\Pi^{max}_{1})\log_2(M-r)		\notag	\\	
	& \le S_{F_{r}}(\Pi_{1})  			\label{eq:24}
\end{align}

Because $S_{F_{r}}(\Pi^{max}_{1})\le S_{F_{r}}(\Pi_{1})$, and $S_{F_{r}}(\cdot)$ decreases monotonically with $\Pi$ (see \emph{theorem \ref{theorem:1}}), there is:
\begin{align}
	[S_{F_{r}}(\Pi^{max}_{1})-S_{F_{r}}(\Pi_{1})](\Pi^{max}_{1}-\Pi_{1}) &\leq 0\\
	\Pi^{max}_{1}-\Pi_{1} & \geq 0\\
	\Pi^{max}_{1} & \geq \Pi_{1} 			
\end{align}

Therefore, we put the real entropy of data set $S$, the number of candidate behaviors of users $M$, and the probability ratio of $r$ behaviors with the highest probability $\{c_{1},c_{2},...,c_{r}\}$ into Eq. (\ref{eq:24}). Thus, the limits of \textit{Top-1} predictability under the data is obtained. We multiply the \textit{Top-1} predictability ($\Pi^{max}_{1}$) by $\sum_{i=1}^{r}c_{i}$, and the limits of \textit{Top-r} predictability is obtained.

\begin{theorem}
	\label{theorem:1}
	The Fano function $S_{F}(\cdot)$ is concave and	monotonically decreases with $p_{i},\{1\leq i \leq r,i \in N^{*}\}$.
\end{theorem}

\begin{proof}
	The $S_{F}(\cdot)$ function is expressed as follows:
	\begin{align}
		S_{F_{r}}(p_{1},...,p_{r}) &= -p_{1}\log_2p_{1}-p_{2}\log_2p_{2}-\cdots-p_{r}\log_2p_{r}		- (1-{\sum}_{i=1}^{r}p_{i})\log_2(1-{\sum}_{i=1}^{r}p_{i}) \notag	\\
		&\quad\,+ (1-{\sum}_{i=1}^{r}p_{i})\log_2(M-r)
	\end{align}
	
	We take the derivative of a single variable, such as $p_{1}$, and the result is as follows:
	\begin{align}
		\frac{\partial S_{F_{r}}(p_{1},...,p_{r})}{\partial p_{1}} &= -(\log_2p_{1}  +\frac{1}{\ln2})+\log_2(1-{\sum}_{i=1}^{r}p_{i})	+\frac{1}{\ln2}-\log_2(M-r)	 \notag	\\
		&= -\log_2(p_{1}/\frac{1-{\sum}_{i=1}^{r}p_{i}}{M-r}) < 0 \label{eq:5}
	\end{align}
	
	Since we scaled the candidate items with the lowest probability $(M-r)$, the probability after scaling was still less than the probability of the $r$ items with the highest probability. We have $p_{i}>\frac{1-{\sum}_{i=1}^{r}p_{i}}{M-r},\{1\leq i \leq r,i \in N^{*}\}$. The above derivation can obtain that the first-order derivative of $S_{F}(\cdot)$ with respect to $p_{i}$ is less than $0$. So we know that $S_{F}(\cdot)$ is decreasing monotonically with respect to $p_{i}$.
	\begin{align}
		\frac{\partial^{2} S_{F_{r}}(p_{1},p_{2},...,p_{r})}{\partial p_{1}^{2}} &= 
		-\frac{1}{p_{1}\ln2}-\frac{1}{(1-{\sum}_{i=1}^{r}p_{i})\ln2} < 0 \label{eq:6}
	\end{align}
	
	We can see from the above that the second derivative of $S_{F}(\cdot)$ with respect to $p_{i}$ is less than $0$, so $S_{F}(\cdot)$ is concave with respect to $p_{i}$.
	
	If we convert $S_{F}(\cdot)$ to a function of one variable and use the formula $p_{i}=c_{i}p_{1}$ to make all $p_{i}$ be expressed by $p_{1}$, the expression of $S_{F}(\cdot)$ will be as follows:
	\begin{align}
		S_{F_{i}}(p_{1}) &= -c_{1}p_{1}\log_2c_{1}p_{1}-c_{2}p_{1}\log_2c_{2}p_{1}-\cdots-		c_{r}p_{1}\log_2c_{r}p_{1}- (1-\sum\limits_{i=1}^{r}c_{i}p_{1})\log_2(1-\sum\limits_{i=1}^{r}c_{i}p_{1}) \notag	\\	
		&\quad\,+ (1-\sum\limits_{i=1}^{r}c_{i}p_{1})\log_2(M-r)
	\end{align}
	
	If we take the derivative of $S_{F}(\cdot)$ for $p_{1}$, we get the following result:
	\begin{align}
		\frac{\partial S'_{F}(p_{1})}{\partial p_{1}} &= 
		-(c_{1}\log_2c_{1}p_{1}+\frac{c_{1}}{\ln2})-(c_{2}\log_2c_{2}p_{1}+\frac{c_{2}}{\ln2})-\cdots-(c_{r}\log_2c_{r}p_{1}+\frac{c_{r}}{\ln2})\notag	\\
		&\quad\,+ {\sum}_{i=1}^{r}c_{i}\log_2(1-{\sum}_{i=1}^{r}c_{i}p_{1}) +\frac{{\sum}_{i=1}^{r}c_{i}}{\ln2}- {\sum}_{i=1}^{r}c_{i}\log_2(M-r) \notag \\
		&=-c_{1}\log_2c_{1}p_{1}-c_{2}\log_2c_{2}p_{1}-\cdots-c_{r}\log_2c_{r}p_{1}	+ {\sum}_{i=1}^{r}c_{i}\log_2(1-{\sum}_{i=1}^{r}c_{i}p_{1})\notag	\\	
		&\quad\,
		- {\sum}_{i=1}^{r}c_{i}\log_2(M-r) \notag \\	
		&=-{\sum}_{i=1}^{r}c_{i}\log_2(c_{i}p_{1}/\frac{1-{\sum}_{i=1}^{r}c_{i}p_{1}}{M-r}) < 0
	\end{align}
	
	In the same way, we have $p_{i}>\frac{1-{\sum}_{i=1}^{r}p_{i}}{M-r},\{1\leq i \leq r,i \in N^{*}\}$, which is $c_{i}p_{1}>\frac{1-{\sum}_{i=1}^{r}p_{i}}{M-r},\{1\leq i \leq r,i \in N^{*}\}$. So we know that $S_{F}(\cdot)$ is monotonically decreasing with respect to $p_{1}$.
	
	We further differentiate $S'_{F}(\cdot)$ with respect to $p_{1}$, and the result is as follows:
	\begin{align}
		\frac{\partial^{2} S_{F}(p_{1})}{\partial p_{1}^{2}} &= -\frac{c_{1}}{p_{1}\ln2}-\frac{c_{2}}{p_{1}\ln2}-\cdots-\frac{c_{r}}{p_{1}\ln2} 
		-\frac{({\sum}_{i=1}^{r}c_{i})^2}{(1-{\sum}_{i=1}^{r}c_{i}p_{1})\ln2} < 0
	\end{align}
	
	Since $c_{i}$ and $p_{1}\ln2$ are both positive, we know that $(1-{\sum}_{i=1}^{r}c_{i}p_{1})$ is greater than $0$. So the second derivative of $S_{F}(\cdot)$ with respect to $p_{1}$, and finally we know that $S_{F}(\cdot)$ is concave.
\end{proof}

\subsection{Theoretical relationship of predictability at different scaling scales}
From the above derivation, it can be seen that when we calculated the predictability of \textit{Top-r}, we also calculated the predictability of \textit{Top-1} in this scaling form. This leads us to think about a question: which one is closer to the real predictability, the \textit{Top-1} predictability calculated by this method, or the \textit{Top-1} predictability calculated by the traditional approach? Next, we will examine the relationship between the calculated predictability and the real predictability under different scaling forms. 
When the probability ratio $\{c_{1},c_{2},\cdot\cdot\cdot,c_{r}\}$ of \textit{Top-r} items is captured, we can obtain the \textit{Top-1} predictability ($\Pi^{max}_{1,r}$) through Eq. (\ref{eq:24}). If we just use the ratio of the probabilities of \textit{Top-(r-1)} items, then we get the new \textit{Top-1} predictability ($\Pi^{max}_{1,r-1}$).


According to Eq. (\ref{eq:24}), it can be known that $S_{F_{r}}(\Pi^{max}_{1,r})\le S_{F_{r}}(\Pi_{1})$, the above formula is still true when $r$ is $r-1$, so we have $S_{F_{r-1}}(\Pi^{max}_{1,r-1}) \le S_{F_{r-1}}(\Pi_{1})$. From the above derivation we know that both $\Pi^{max}_{1,r}$ and $\Pi^{max}_{1,r-1}$ are limits on predictability. $S_{F_{r}}(\cdot)$ and $S_{F_{r-1}}(\cdot)$ are concave functions and decrease monotonically with $\Pi$. And we have the following formula:
\begin{align}
	S_{F_{r}}(\Pi^{max}_{1,r}) &=S\left(\Pi^{max}_{1,r},\cdot\cdot\cdot,c_{r}\Pi^{max}_{1,r},\frac{1-\sum\limits_{i=1}^{r}c_{i}\Pi^{max}_{1,r}}{M-r},\cdot\cdot\cdot,\frac{1-\sum\limits_{i=1}^{r}c_{i}\Pi^{max}_{1,r}}{M-r}\right) \notag\\ 
	&\leq S\left(\Pi^{max}_{1,r},\cdot\cdot\cdot,c_{r-1}\Pi^{max}_{1,r},\frac{1-\sum\limits_{i=1}^{r-1}c_{i}\Pi^{max}_{1,r}}{M-(r-1)},\cdot\cdot\cdot,\frac{1-\sum\limits_{i=1}^{r-1}c_{i}\Pi^{max}_{1,r}}{M-(r-1)}\right) \notag\\
	&= S_{F_{r-1}}(\Pi^{max}_{1,r})
\end{align}


Since the entropy of $S_{F_{r}}(\Pi^{max}_{1,r})$ and $S_{F_{r-1}}(\Pi^{max}_{1,r})$ in the previous $(r-1)$ part is the same, the entropy of $S_{F_{r-1}}(\Pi^{max}_{1,r})$ is greater than $S_{F_{r}}(\Pi^{max}_{1,r})$ in the remaining part of $M-(r-1)$, so $S_{F_{r}}(\Pi^{max}_{1,r}) \le S_{F_{r-1}}(\Pi^{max}_{1,r})$.
And then we know that  $S=S_{F_{r-1}}(\Pi^{max}_{1,r-1})= S_{F_{r}}(\Pi^{max}_{1,r})$. So we can get:
\begin{align}
	S_{F_{r-1}}(\Pi^{max}_{1,r-1})= S_{F_{r}}(\Pi^{max}_{1,r})\leq S_{F_{r-1}}(\Pi^{max}_{1,r})
\end{align}

And since the $S_{F_{r-1}}(\cdot)$ function is monotonically decreasing with $\Pi$, so $\Pi^{max}_{1,r} \le \Pi^{max}_{1,r-1}$. Repeating the above derivation, we can get the following results:
\begin{align}
	\Pi_{1} \leq \Pi^{max}_{1,r} \leq \Pi^{max}_{1,r-1} \leq \cdot\cdot\cdot \leq \Pi^{max}_{1,1}
\end{align}

Therefore, we can know that the method derived from $\{c_{1},c_{2},\cdot\cdot\cdot,c_{r}\}$ to obtain the predictability of \textit{Top-r} can not only obtain the predictability of \textit{Top-r}, but also obtain the predictability of \textit{Top-1} which is closer to the true value than that obtained from the traditional method. The deviation between the predictability of \textit{Top-r} and the true predictability decreases with the increase of $r$.



\section{Experiments}
In this part, we first introduce the real-world datasets used in the paper in Section \ref{experiment:1}. In Section \ref{ten_alg}, we present the algorithms that are popular for \textit{Top-N} recommendations. In Section \ref{experiment:3}, we analyze the pattern obeyed by the $N$ items with the highest probability to obtain the probability relationship between the \textit{Top-N} items. 
Section \ref{experiment:4}, we introduce how to generate data with specified predictability to set the stage for later experiments under generated data. 
In Section \ref{experiment:5} our experiments on the generated data compare the measurement bias of predictability of existing methods with that of our approach to show that our effect is greatly improved. Also, in this section, we quantify the measure biases under different scaling forms to derive more accurate predictability and make the measure biases in a table for easy access by others. 
Finally, in Section \ref{experiment:6}, we compare the accuracy achieved by existing algorithms with the predictability calculated by our method to verify whether our approach can capture the variation in the dataset's prediction difficulty and understand the room for improvement of the existing accuracy.

\subsection{Datasets} \label{experiment:1}
We conduct experiments on four real-world datasets which record the user purchase histories with the details as below:
\begin{itemize}
	\item DUNN: A dataset contains offline transactions from about 2,000 households over two years $\footnote{\url{https://www.dunnhumby.com/source-files/}}$
	\item INSTA: This data set was published by Instacart, a website that offers same-day grocery delivery in the US $\footnote{\url{https://www.kaggle.com/c/instacart-market-basket-analysis/data}}$. It contains more than 3 million grocery orders from more than 200,000 users. There is no specific date for each order, but the order of transactions for each user is provided.
	\item RSC15: It was released in the ACM RECSYS 2015 Challenge, which included records of users' online purchases and clicks sequences over six months $\footnote{\url{https://www.dropbox.com/sh/7qdquluflk032ot/AACoz2Go49q1mTpXYGe0gaANa?dl=0}}$.
	\item TMALL: This dataset was released in the context of the TMALL contest and contained a year's worth of interactive logs from TMALL $\footnote{\url{https://tianchi.aliyun.com/dataset/dataDetail?dataId=42}}$.
\end{itemize}
\begin{table}
	\caption{Basic characteristics of the four datasets.}
	\label{tab_dataset}
	\centering
	\begin{tabular}{lcccc}
		\toprule
		Dataset &\#user		&\#item 		&\#transaction & \#timespan\\
		\midrule
		Dunnhumby 	&2,500  	&26,780  	&269,974 &2 years \\
		Instacart 	&206,120  	&42,987 	&3,345,786 &30 days  \\
		RSC15 		&186,600	&28,582  	&5,426,961 &182 days \\
		TMALL 		&131,450	&425,348  	&13,418,695 &90 days \\
		\bottomrule
	\end{tabular}
\end{table}

\subsection{Top-N recommendation algorithms used in the study} \label{ten_alg}
To evaluate the performance of our predictability method, for each of the above datasets, we calculate its \textit{Top-N} accuracy limits, denoted as $V_u$, and compare it to the best of the $M$ well-performed recommendation algorithms' results, say $V_r = \argmax\{R_i\}$ where $R_i$ the \textit{Top-N} accuracy produced by the $i^{th}$ recommendation algorithm. The closer these two values  $V_u$ and $V_r$, the better our predictability method performs. In this study, we set $M=10$, and the details of the recommendation algorithms are as below:
\begin{itemize}	
	\item \emph{Simple Association Rules}: AR \cite{2001Mining} is a simplified version of the association rule mining technique. The method is designed to capture the frequency of two co-occurring events. 
	\item \emph{Markov Chains (MC)}: MC \cite{2004General} can be considered as a variant of AR that focuses on transformational relationships in sequential data.  
	\item \emph{Item-based KNN (IKNN)}: The IKNN \cite{barkan2016item2vec} method considers only the last item in a given sequence and then returns the item with the highest frequency of co-occurrence with that item in the data as a recommendation.
	\item \emph{Sequential Rules (SR)}: The SR \cite{2008An} is a variant of MC or AR, respectively. It also considers the order of behavior but is less restrictive.
	\item \emph{Factorized Personalized Markov Chains (FPMC)} \cite{rendle2010factorizing}: In order to implement recommendation scenarios considering user-taste and sequential information, this algorithm models a markov transition matrix for each user to fuse both sequential and personalized information.
	\item \emph{Factored Item Similarity Models (FISM)}: FISM \cite{2017Factored} belongs to the classical approach of \textit{Top-N} recommendation algorithms, which extends the item-based hidden factor approach to \textit{Top-N} problems. 
	\item \emph{Factorized Sequential Prediction with Item Similarity Models (FOSSIL)} \cite{he2016fusing}: This method combines FISM with decomposable Markov Chains to integrate sequence information into the model.
	\item \emph{Session-based Matrix Factorization}: SMF \cite{shi2013mining} is a factorization-based recommendation model designed for serialization-based recommendation tasks. 
	\item \emph{Bayesian Personalized Ranking (BPR)}: BPR-MF \cite{2018Recommendation} is a learning ranking method for implicit feedback recommendation scenarios. 
	\item \emph{Gru4Rec}: Gru4Rec \cite{hidasi2015session} models user behavior sequences by RNN with gated recurrent units to learn the evolution patterns between pre and post behaviors.
\end{itemize}

\begin{figure*}[t]
	\centering
	\setlength{\abovecaptionskip}{0pt}
	\setlength{\belowcaptionskip}{0pt}
	\includegraphics[width=0.98\textwidth]{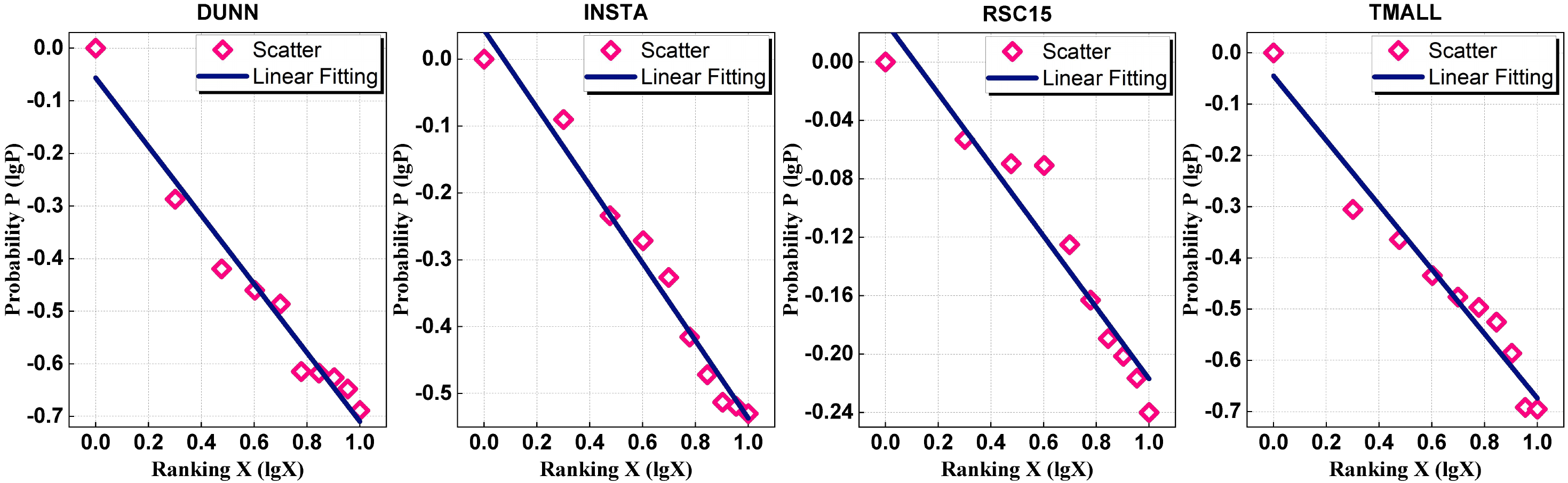} 
	\caption{The probabilistic proportional distribution of the user's favorite items. The frequency of an item's occurrence is inversely proportional to its ranking in the frequency table.}
	\label{fig_ranking}
	\vspace{0pt}
\end{figure*}
\subsection{Distribution of Shopping Frequency}\label{experiment:3}
We counted the historical behavior of the group under the whole dataset and calculated the number of interactions of the items in order from the largest to the smallest. After the frequency is obtained, it is uniformly divided by the highest frequency for normalization so as to get the $\{c_{1},c_{2},...,c_{r}\}$ of the population. The patterns of human shopping behavior are quite uneven, so the frequency $f_{k}$ of the $k^{th}$ most-interacted item conforms to Zipf's law \cite{kluckhohn1950human}. The distribution is shown in Fig. \ref{fig_ranking}, and the formula is as follows:
\begin{align}
	f_{k}\backsim k^{-\xi}
\end{align}

The experimental results show that $\xi=0.6\pm 0.07$ on the four classical datasets in the paper. This indicates that the frequency distribution of user interaction behavior is in line with $P(f)\backsim f^{-(1+1/\xi)}$. When $r$ is greater than a certain number, there is a large gap between the value of $c_{r}$ and $c_{1}$, which will cause the deviation of the predictability of \textit{Top-1} to be passed on multiple times. This will put a high demand on the quantization precision of $c_{r}$. Therefore, in order to avoid this situation, we chose to intercept $r$ as $10$, $c_{10}$ is around $0.25$, so that the value difference between $c_{i}$ is not big.


\subsection{Generates Data with Specified Predictability} \label{experiment:4}
\subsubsection{The first generation method}
\begin{figure}[t]
	\centering
	\subfloat{
		\includegraphics[width=0.35\textwidth]{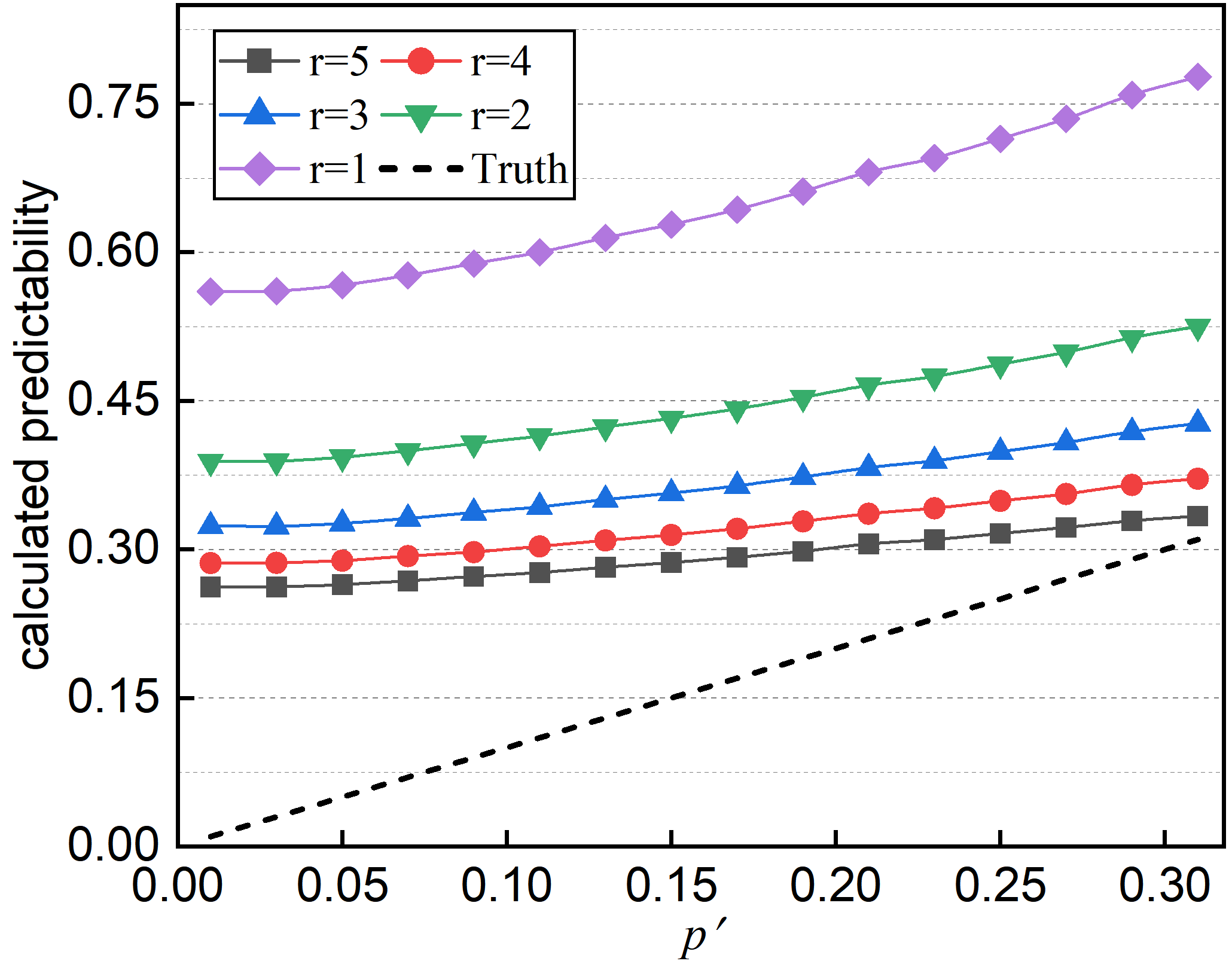}
	}
	\hspace{5mm}
	\subfloat{
		\includegraphics[width=0.35\textwidth]{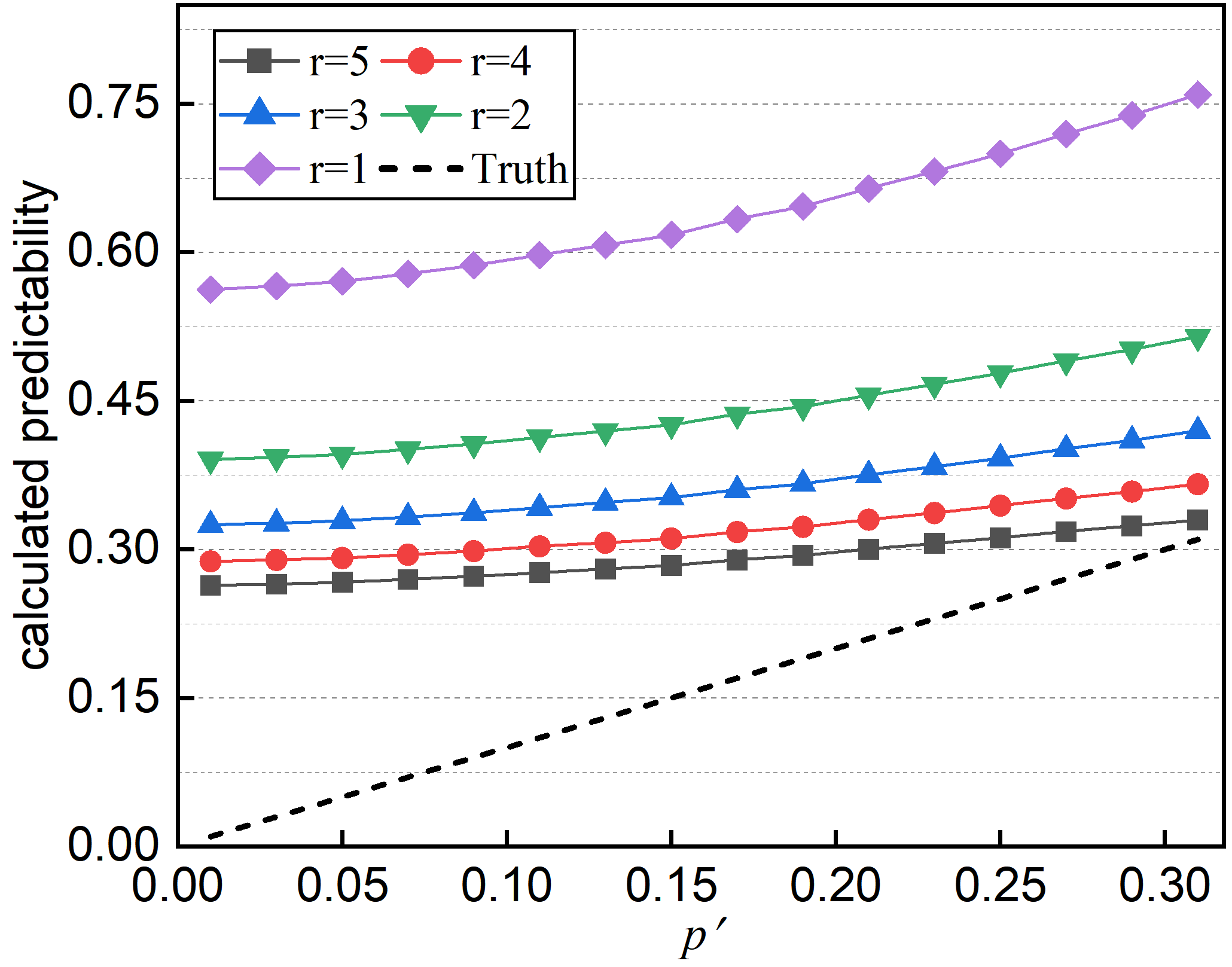}
	}
	\caption{The results of the predictability evaluation under the generated data}
	\label{fig_GD}
\end{figure}
To truly demonstrate the effectiveness of our method from experimental data, we generate behavior sequences with known predictability. An effective estimate of the predictability of the sequence can indicate the accuracy of the work.
We generate the corresponding Markov sequences \cite{1957Journal} and regulate the predictability of the sequences by adjusting the transition probability $p$. This generates a sequence of behaviors whose predictability is known. For example, we set $p=0.2$, and $c_{1}=1,c_{2}=0.7,c_{3}=0.6,c_{4}=0.5$. Therefore, the predictability of the four behaviors with the highest probability follows $\Pi_{1}=0.2,\Pi_{2}=0.14,\Pi_{3}=0.12,\Pi_{4}=0.1$. The predictability of $\{\textup{Top-1},\textup{Top-2},\textup{Top-3},\textup{Top-4}\}$ are $0.2, 0.34, 0.46, 0.56$, respectively. 
We consider a simple generative rule where the next moment state is only relevant to the current state, according to the following Markovian transfer matrix
\begin{equation}
	\label{TransfetMatrix}
	\begin{blockarray}{ccccccc}
		&A&B&C&D&E&R\\
		\begin{block}{c[cccccc]}
			A&	c_{1}p&	c_{2}p&	c_{3}p&	c_{4}p&	c_{5}p&	1-\sum_{i=1}^{5}c_{i}p\\
			B&	1-\sum_{i=1}^{5}c_{i}p&	c_{1}p&	c_{2}p&	c_{3}p&	c_{4}p&	c_{5}p\\
			C&	c_{5}p&	1-\sum_{i=1}^{5}c_{i}p&	c_{1}p&	c_{2}p&	c_{3}p&	c_{4}p\\
			D&	c_{4}p&	c_{5}p&	1-\sum_{i=1}^{5}c_{i}p&	c_{1}p&	c_{2}p&	c_{3}p\\
			E&	c_{3}p&	c_{4}p&	c_{5}p&	1-\sum_{i=1}^{5}c_{i}p&	c_{1}p&	c_{2}p\\
			R&	c_{2}p&	c_{3}p&	c_{4}p&	c_{5}p&	1-\sum_{i=1}^{5}c_{i}p&	c_{1}p\\
		\end{block}
	\end{blockarray}.
\end{equation}

First, we set the number of states to be $M$, so the number of random states represented by $R$ is $M-5$. We set $\xi=0.6$ to determine the relationship between $c_{i}$. We repeat this operation $L$ times, resulting in a behaviors sequence of length $L$ with known predictability.
The Figure \ref{fig_GD} on the left shows the results of the predictability method under the first generation method. The predictability of the data is changed by adjusting $p^{'}$, which is $c_{1}p$, where the black dashed line represents the true value, $r=1$ is the result of the existing method $\Pi^{max}_{1,1}$, and $r=5$ represents the predictability result $\Pi^{max}_{1,5}$ obtained by our method which incorporates the \textit{Top-5} item relationship into the scaling. The figure also shows the results for other values of r, $r={2,3,4}$. From this it can be seen that the grey curve corresponding to $r=5$ is much more forced to the dashed line and works significantly better. The Figure \ref{fig_length} on the left shows the effect of the length of the sequence on the method. It can be seen that our method stabilises at $2^9$, while the existing method stabilises at $2^{11}$, and we have a more relaxed requirement for the sequence length to get good evaluation results earlier.

\subsubsection{The second generation method}
\begin{figure}[t]
	\centering
	\subfloat{
		\includegraphics[width=0.35\textwidth]{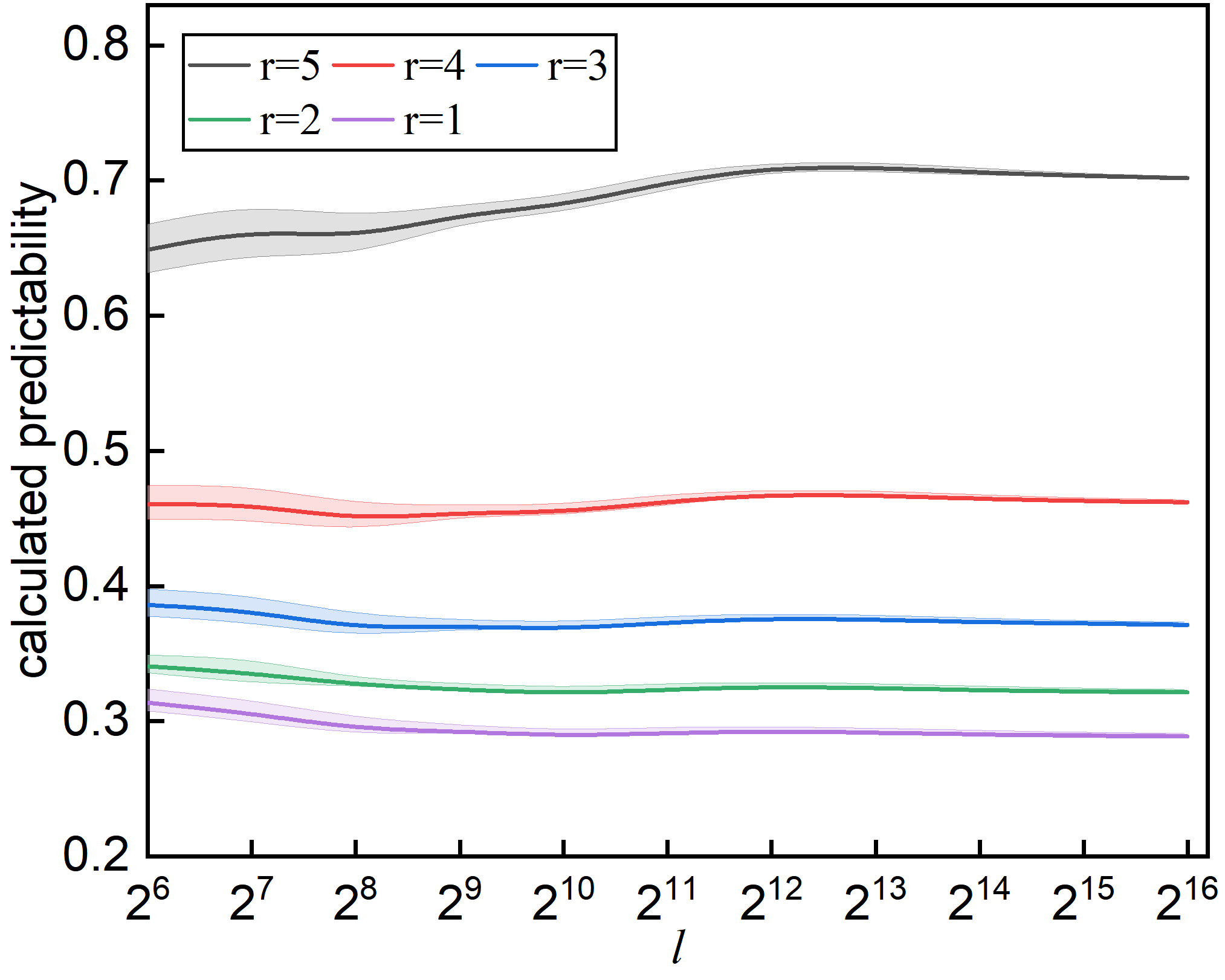}
	}
	\hspace{5mm}
	\subfloat{
		\includegraphics[width=0.35\textwidth]{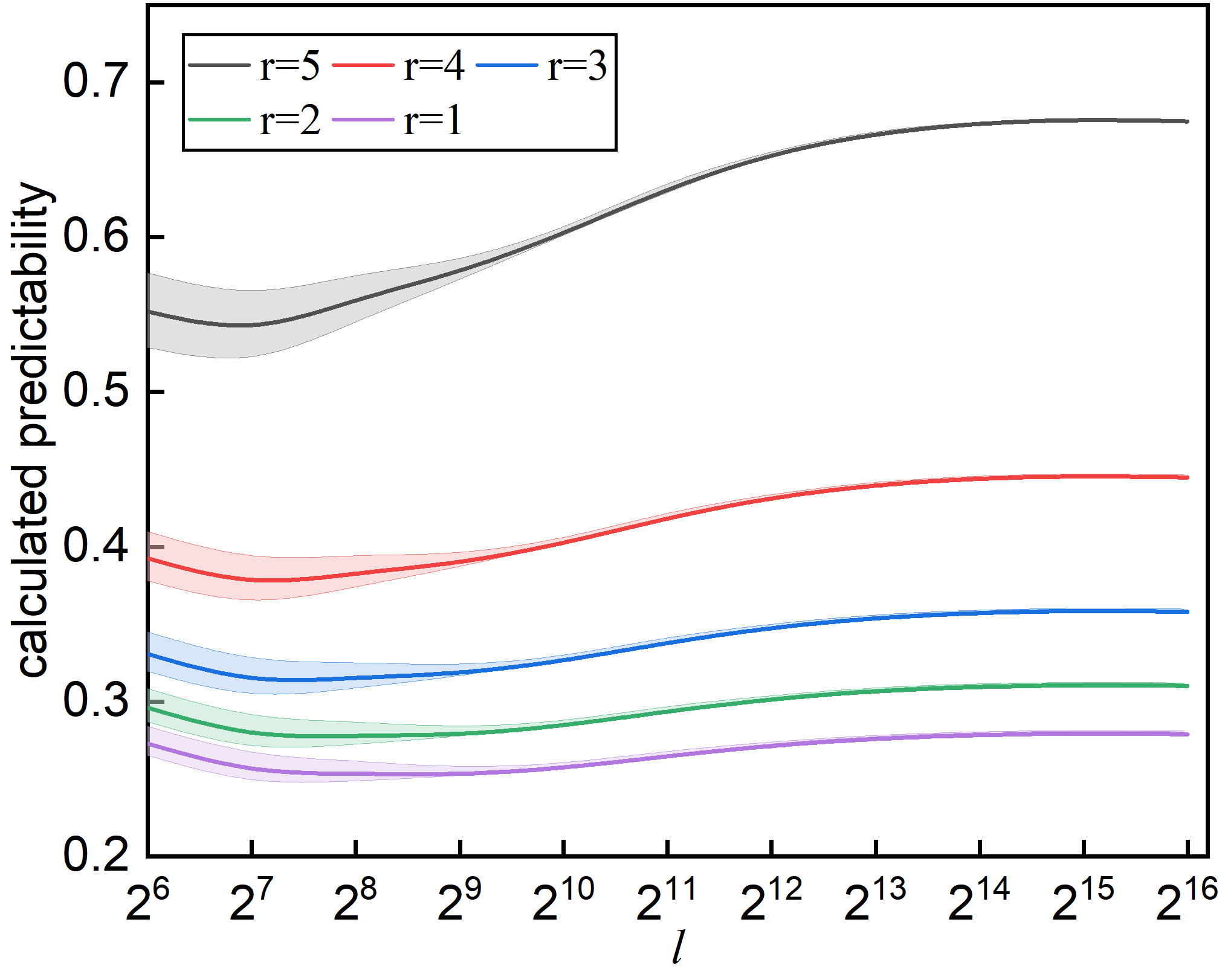}
	}
	\caption{The effect of sequence length on predictability methods}
	\label{fig_length}
\end{figure}
We consider a more complex generation where the next moment state is no longer a simple probabilistic transition of the current state, but depends on the previous two states. We assume that the set of states is $\Omega = \{S_1,S_2,\cdots,S_M\}$, the state at time $t-1$ is $\omega^{t-1}=S_i$, the current state is $\omega^t=S_j$, and the set of probabilities is $p=[c_{1}p,c_{2}p,c_{3}p,c_{4}p,c_{5}p]$, then the set of cumulative probabilities is $p=[c_{1}p,\sum_{i=1}^{2}c_{i}p,\sum_{i=1}^{3}c_{i}p,\sum_{i=1}^{4}c_{i}p,\sum_{i=1}^{5}c_{i}p]$. We randomly generate the next moment state with probability $p$ between 0 and 1 when $p>\sum_{i=1}^{5}c_{i}p$. When $\sum_{i=1}^{x}c_{i}p< p\leq\sum_{i=1}^{x+1}c_{i}p$, the next moment state is $\omega^{t+1}=S_k$, $k=i+j+x$ (if $k>M$, we set $k \leftarrow k-M$).
Obviously, the true \textit{Top-1} predictability is $T=c_{1}p+(1-\sum_{i=1}^{5}c_{i}p)/M\approx c_{1}p$.

In the second generation method, we also set $\xi=0.6$. The Figure \ref{fig_GD} on the right shows the results of the predictability method under the second generation method. The curve represents the same meaning as the figure on the left. We can see that our method still shows great advantages in the second generation method. The Figure \ref{fig_length} on the right shows the effect of the length of the sequence on the method. It can be seen that our method stabilises at $2^{11}$, while the existing method stabilises at $2^{13}$. All methods in the second generation method have more stringent requirements on length, but our method still has relatively relaxed requirements and reaches the steady state faster.

\begin{table*}
	\caption{Measurement deviation under different scaling forms on a sequence with the \textit{Top-1} predictability of $0.2$. $S_{F_{1}}$ represents the existing method, and the rest is our proposed scaling form.}
	\label{tab:1}
	\centering
	\begin{tabular}{lcccccccccc}
		\toprule
		&$S_{F_{1}}$	&$S_{F_{2}}$	&$S_{F_{3}}$	&$S_{F_{4}}$	&$S_{F_{5}}$	&$S_{F_{6}}$	&$S_{F_{7}}$	&$S_{F_{8}}$	&$S_{F_{9}}$	&$S_{F_{10}}$ \\
		\midrule
		Deviation	&261\%	&132\%	&85\%	&58\%	&41\%	&29\%	&20\%	&13\%	&7\%	&2\% \\
		\bottomrule
	\end{tabular}
\end{table*}

\subsection{Quantify the Predictability Measure Bias}\label{experiment:5}
Based on the real-world dataset results, we set $p=0.2, \xi=0.6$ for the hyperparameters of the above generation sequence method. Thus, sequences with known predictability are generated. By our method, the predictability of $S_{F_{i}}$ for different scaling forms can be obtained. The results are shown in Table \ref{tab:1}, and the deviation of predictability calculated by Song's method is 261\%. As the scaling part gets smaller, that is, the $i$ in $S_{F_{i}}$ keeps increasing, the measurement deviation gets smaller. The minimum measurement deviation is 2\%, which is a considerable improvement.

To further understand the true predictability, we adjusted $p$ and $\xi$ to calculate the predictability under different scenarios and the corresponding deviation. According to the previous experiment, the value of $\xi$ in the recommended scenario is $0.6\pm 0.07$. Finally, we set the $p\in \{0.01,0.02,...,0.62\},\xi\in \{0.53,0.54,...,0.67\}$, and the subscript $i$ of $S_{F_{i}}$ follows $i\in \{1,2,...,10\}$. We calculated the predictability bias at various scales for direct review by other researchers $\footnote{\url{https://drive.google.com/drive/folders/1sEFufZHiyuhd0d-Lgz1VmbokK7_lZ0LB}}$. We can get the deviation between $\Pi^{max}_{1}$ and $\Pi_{1}$ by querying the table after the calculated predictability, and $\xi$ is obtained, to estimate the value closer to the real predictability. After the re-estimation, the new predictability is closer to the real predictability.

\begin{table*}[!htb]
	\caption{The best accuracy performance of the ten recommendation algorithms in the evaluation indexes from \textit{Top-1} to \textit{Top-10}.}
	\label{tab:2}
	\centering
	\begin{tabular}{lcccccccccc}
		\toprule
		Dataset &Top-1	&Top-2	&Top-3	&Top-4	&Top-5	&Top-6	&Top-7	&Top-8	&Top-9	&Top-10\\
		\midrule
		DUNN	&0.0953	&0.1473	&0.1846	&0.2128	&0.2361	&0.2560	&0.2734	&0.2881	&0.3017	&0.3150	\\
		INSTA	&0.0378	&0.0645	&0.0858	&0.1039	&0.1198	&0.1326	&0.1440	&0.1554	&0.1655	&0.1752	\\
		RSC15 	&0.2252	&0.3268	&0.3834	&0.4255 &0.4517	&0.4759	&0.4933	&0.5072	&0.5184 &0.5311	\\
		TMALL	&0.2248	&0.2752	&0.2943	&0.3038	&0.3099	&0.3146	&0.3184	&0.3213	&0.3240	&0.3261	\\
		\bottomrule
	\end{tabular}
\end{table*}
\begin{table*}[!htb]
	\caption{The results of predictability under the evaluation indicators \textit{Top-1} to \textit{Top-10} in four real-world datasets.}
	\label{tab:3}
	\centering
	\begin{tabular}{lcccccccccc}
		\toprule
		Dataset &Top-1	&Top-2	&Top-3	&Top-4	&Top-5	&Top-6	&Top-7	&Top-8	&Top-9	&Top-10 \\
		\midrule
		DUNN	&0.2185	&0.3335	&0.4175	&0.4951	&0.5655	&0.6181	&0.6702	&0.7222	&0.7706	&0.8148	\\
		INSTA	&0.1690	&0.3043	&0.3991	&0.4850	&0.5613	&0.6241	&0.6787	&0.7294	&0.7796	&0.8287	\\
		RSC15 	&0.1852	&0.3211	&0.4246	&0.5196	&0.5971	&0.6629	&0.7257	&0.7878	&0.8476	&0.9051	\\
		TMALL	&0.1810	&0.2893	&0.3905	&0.4773	&0.5602	&0.6260	&0.6737	&0.7196	&0.7613	&0.7997	\\
		\bottomrule
	\end{tabular}
\end{table*}

\subsection{Performance of Algorithms on Real-World Datasets}\label{experiment:6}
We selected 10 popular algorithms of the \textit{Top-N} recommender system and obtained the highest accuracy of the algorithms under four datasets. These ten algorithms are the ones we introduced in Section \ref{ten_alg}.
The best performance results of the $10$ algorithms under each accuracy evaluation index are shown in Table \ref{tab:2} (See the detailed results in \emph{Appendix, The detailed results of predictability on real-world datasets}). At the same time, we calculate the predictability of \textit{Top-N} under each data set. First, we calculated the \textit{Top-1} predictability through Eq. (\ref{eq:24}) then looked up the table to get the scaling deviation, and further obtained the more accurate \textit{Top-1} predictability. 
Finally, we calculated the predictability from \textit{Top-1} to \textit{Top-10} respectively according to the $\{c_{1},c_{2},\cdot\cdot\cdot,c_{r}\}$. The results are shown in Table \ref{tab:3}.

\section{Conclusion}
Throughout history, scientists in all eras have tried to predict the future. Newtonian mechanics Astrophysics is committed to studying the motion of objective things to predict the future state of the object. The continuous expansion of data on large-scale human consumption behavior and the development of theoretical models have enabled scientists to analyze further and understand human behavior and thus to realize the prediction of behavior. Predictability studies explore the limits of accuracy. The research on the predictability of recommender systems is to explore the limits of the regularity of human consumption behavior. Since there are a large number of scenarios in the recommender system that need to recommend multiple items at the same time, the accuracy of \textit{Top-N} is widely used as an indicator to evaluate the quality of the recommendation algorithm. However, the existing theory of predictability cannot deduce the predictability of \textit{Top-N} recommendation. 
We successfully constructed the predictability ratio of \textit{Top-N} behaviors with the highest probability, thus achieving the quantification of the \textit{Top-N} predictability and theoretically proving that our method is more accurate than the existing theories. The high precision limits of predictability dramatically enhances the practical significance of predictability research. 
The frequency of user purchases in the datasets investigated in this paper obeys Zipf's law, and $\xi$ is around $0.6$. However, the recommender system involves many scenarios, including shopping, music, news, video recommendation, etc. $\xi$ may not be within the range of our statistics, but our method is still applicable. However, since the scope involved is too large, we did not count the scaling deviation in these cases. This requires the user to calculate the specific deviation according to our method and then complete the calculation of the \textit{Top-N} predictability.

\section*{Acknowledgments}
This work was supported in part by the National Natural Science Foundation of China (No. 61960206008, No. 62002294) and the National Science Fund for Distinguished Young Scholars (No. 61725205).

\bibliographystyle{unsrt}  
\bibliography{TOP-N}

\appendix
\section*{Appendix A The detailed results of predictability on real-world datasets}
In the main manuscript, we only show the best accuracy performance of the ten recommendation algorithms in the evaluation indexes from \textit{Top-1} to \textit{Top-10} of the five real-world datasets. Here we present the specific performance results of $10$ recommendation algorithms for each metric.

\begin{table*}
	\caption{Detailed performance results of $10$ recommendation algorithms in real-world datasets under the evaluation indexes \textit{Top-1} to \textit{Top-10}.}
	\label{tab:1}
	\centering
	\begin{tabular}{cl|p{0.9cm}<{\centering}p{0.9cm}<{\centering}p{0.9cm}<{\centering}p{0.9cm}<{\centering}p{0.9cm}<{\centering}p{0.9cm}<{\centering}p{0.9cm}<{\centering}p{0.9cm}<{\centering}p{0.9cm}<{\centering}p{0.9cm}<{\centering}}
		\toprule 
		&\textbf{Dataset}&\textbf{AR}&\textbf{SMF}&\textbf{FISM}&\textbf{BPR}&\textbf{MC}&\textbf{SR}&\textbf{FOSSIL}&\textbf{Gru4Rec}&\textbf{IKNN}&\textbf{FPMC} \\
		\hline
		\multirowcell{5}{\textbf{Top-1}}
		&DUNN &0.0350 &0.0347 &0.0860 &0.0799 &0.0000 &0.0090 &0.0953 &0.0226 &0.0058 &0.0249 \\
		&INSTA &0.0312 &0.0378 &0.0317 &0.0345 &0.0036 &0.0338 &0.0014 &0.0018 &0.0182 &0.0215 \\
		&RSC15 &0.1420 &0.0985 &0.1996 &0.1948 &0.2229 &0.2252 &0.1921 &0.0538 &0.1076 &0.2067 \\
		&TMALL &0.1164 &0.1276 &0.0877 &0.0181 &0.1667 &0.1285 &0.0071 &0.0897 &0.2248 &0.1110 \\
		
		\hline
		\multirowcell{5}{\textbf{Top-2}}
		&DUNN &0.0554 &0.0590 &0.1282 &0.1247 &0.0014 &0.0182 &0.1473 &0.0460 &0.0092 &0.0444 \\
		&INSTA &0.0556 &0.0645 &0.0561 &0.0589 &0.0151 &0.0580 &0.0044 &0.0058 &0.0323 &0.0391 \\
		&RSC15 &0.2193 &0.1547 &0.2981 &0.2992 &0.3268 &0.3257 &0.2605 &0.0906 &0.1275 &0.2777 \\
		&TMALL &0.1353 &0.1504 &0.0916 &0.0300 &0.1822 &0.1540 &0.0093 &0.1238 &0.2752 &0.1346 \\
		
		\hline
		\multirowcell{5}{\textbf{Top-3}}
		&DUNN &0.0727 &0.0796 &0.1605 &0.1556 &0.0034 &0.0278 &0.1846 &0.0680 &0.0119 &0.0571 \\
		&INSTA &0.0755 &0.0858 &0.0752 &0.0800 &0.0288 &0.0779 &0.0068 &0.0117 &0.0446 &0.0544 \\
		&RSC15 &0.2657 &0.2021 &0.3595 &0.3602 &0.3834 &0.3790 &0.2879 &0.1176 &0.1382 &0.3052 \\
		&TMALL &0.1450 &0.1632 &0.0944 &0.0374 &0.1872 &0.1683 &0.0120 &0.1416 &0.2943 &0.1489 \\
		
		\hline
		\multirowcell{5}{\textbf{Top-4}}
		&DUNN &0.0896 &0.0960 &0.1848 &0.1787 &0.0059 &0.0363 &0.2128 &0.0872 &0.0144 &0.0673 \\
		&INSTA &0.0931 &0.1039 &0.0932 &0.0976 &0.0413 &0.0950 &0.0081 &0.0187 &0.0556 &0.0678 \\
		&RSC15 &0.3043 &0.2349 &0.4062 &0.4067 &0.4255 &0.4166 &0.3035 &0.1389 &0.1473 &0.3255 \\
		&TMALL &0.1503 &0.1723 &0.0975 &0.0437 &0.1906 &0.1791 &0.0139 &0.1545 &0.3038 &0.1624 \\
		
		\hline
		\multirowcell{5}{\textbf{Top-5}}
		&DUNN &0.1012 &0.1117 &0.2058 &0.1993 &0.0087 &0.0440 &0.2361 &0.1035 &0.0167 &0.0767 \\
		&INSTA &0.1082 &0.1198 &0.1073 &0.1141 &0.0524 &0.1104 &0.0089 &0.0249 &0.0654 &0.0786 \\
		&RSC15 &0.3323 &0.2527 &0.4379 &0.4383 &0.4517 &0.4422 &0.3123 &0.1580 &0.1529 &0.3408 \\
		&TMALL &0.1571 &0.1803 &0.0995 &0.0488 &0.1933 &0.1880 &0.0165 &0.1658 &0.3099 &0.1706 \\
		
		\hline
		\multirowcell{5}{\textbf{Top-6}}
		&DUNN &0.1099 &0.1252 &0.2215 &0.2159 &0.0115 &0.0523 &0.2560 &0.1191 &0.0191 &0.0866 \\
		&INSTA &0.1215 &0.1326 &0.1202 &0.1281 &0.0634 &0.1237 &0.0097 &0.0322 &0.0754 &0.0895 \\
		&RSC15 &0.3567 &0.2688 &0.4647 &0.4661 &0.4759 &0.4631 &0.3177 &0.1765 &0.1583 &0.3499 \\
		&TMALL &0.1603 &0.1856 &0.1010 &0.0529 &0.1954 &0.1951 &0.0195 &0.1749 &0.3146 &0.1772 \\
		
		\hline
		\multirowcell{5}{\textbf{Top-7}}
		&DUNN &0.1197 &0.1388 &0.2352 &0.2315 &0.0142 &0.0592 &0.2734 &0.1324 &0.0216 &0.0940 \\
		&INSTA &0.1345 &0.1440 &0.1319 &0.1408 &0.0733 &0.1357 &0.0106 &0.0379 &0.0832 &0.0994 \\
		&RSC15 &0.3714 &0.2843 &0.4831 &0.4895 &0.4933 &0.4765 &0.3231 &0.1957 &0.1639 &0.3560 \\
		&TMALL &0.1627 &0.1908 &0.1022 &0.0565 &0.1971 &0.2018 &0.0221 &0.1824 &0.3184 &0.1823 \\
		
		\hline
		\multirowcell{5}{\textbf{Top-8}}
		&DUNN &0.1292 &0.1498 &0.2486 &0.2458 &0.0173 &0.0647 &0.2881 &0.1437 &0.0235 &0.1012 \\
		&INSTA &0.1451 &0.1554 &0.1425 &0.1519 &0.0825 &0.1461 &0.0119 &0.0434 &0.0915 &0.1088 \\
		&RSC15 &0.3853 &0.2960 &0.5004 &0.5074 &0.5071 &0.4863 &0.3268 &0.2109 &0.1692 &0.3611 \\
		&TMALL &0.1655 &0.1947 &0.1033 &0.0601 &0.1984 &0.2070 &0.0240 &0.1891 &0.3213 &0.1868 \\
		
		\hline
		\multirowcell{5}{\textbf{Top-9}}
		&DUNN &0.1377 &0.1599 &0.2598 &0.2594 &0.0208 &0.0704 &0.3017 &0.1546 &0.0253 &0.1075 \\
		&INSTA &0.1565 &0.1655 &0.1524 &0.1619 &0.0926 &0.1562 &0.0135 &0.0487 &0.0994 &0.1170 \\
		&RSC15 &0.3961 &0.3052 &0.5125 &0.5217 &0.5184 &0.4932 &0.3300 &0.2248 &0.1720 &0.3664 \\
		&TMALL &0.1672 &0.1988 &0.1042 &0.0636 &0.1997 &0.2120 &0.0253 &0.1953 &0.3240 &0.1905 \\
		
		\hline
		\multirowcell{5}{\textbf{Top-10}}
		&DUNN &0.1440 &0.1701 &0.2709 &0.2704 &0.0240 &0.0758 &0.3150 &0.1664 &0.0269 &0.1142 \\
		&INSTA &0.1668 &0.1752 &0.1626 &0.1714 &0.1016 &0.1652 &0.0158 &0.0533 &0.1069 &0.1250 \\
		&RSC15 &0.4067 &0.3146 &0.5263 &0.5347 &0.5283 &0.5005 &0.3326 &0.2343 &0.1755 &0.3700 \\
		&TMALL &0.1694 &0.2022 &0.1051 &0.0667 &0.2010 &0.2169 &0.0270 &0.2032 &0.3261 &0.1944 \\
		
		\bottomrule
	\end{tabular}
\end{table*}

\end{document}